\documentclass[11pt]{article}
\usepackage[utf8]{inputenc}
\usepackage[
 a4paper,
 total={160mm,257mm},
 left=25mm,
 top=20mm]{geometry}
\usepackage[pagebackref]{hyperref}
\usepackage{amsfonts}
\usepackage{amsmath}
\usepackage{amssymb}
\usepackage{mathrsfs}  
\usepackage{xcolor}
\usepackage{authblk}
\usepackage{physics}
\usepackage{multicol}
\usepackage{tikz}
\usetikzlibrary{cd}
\usepackage{setspace}
\usepackage{amsthm}
\usepackage{mathtools}
\usepackage{tensor}

\definecolor{c1}{RGB}{0,90,187}
\definecolor{c2}{RGB}{255,213,0}
\hypersetup{colorlinks,breaklinks,
            linkcolor = purple,
            citecolor = teal,
            urlcolor = purple,
            }

\theoremstyle{definition}
\newtheorem{theorem}{Theorem}[section]
\newtheorem{corollary}[theorem]{Corollary}
\newtheorem{definition}[theorem]{Definition}
\newtheorem{lemma}[theorem]{Lemma}
\newtheorem{proposition}[theorem]{Proposition}

\theoremstyle{remark}

\DeclareMathOperator{\I}{i}
\DeclareMathOperator{\diff}{d}
\DeclareMathOperator{\Max}{Max}

\makeatletter
\newcommand{\doublewidetilde}[1]{{%
  \mathpalette\double@widetilde{#1}%
}}
\newcommand{\double@widetilde}[2]{%
  \sbox\z@{$\m@th#1\widetilde{#2}$}%
  \ht\z@=.9\ht\z@
  \widetilde{\box\z@}%
}

\numberwithin{equation}{section}

\title{\textbf{Structure of deformed $w_{1+\infty}$ symmetry and topological generalization in Celestial CFT }} 
\author[1]{\textsc{Pavel Drozdov}}
\author[2]{\textsc{Taro Kimura}}

\affil[1]{Universit\`a degli Studi di Udine, Italy}

\affil[2]{%
Institut de Math\'ematiques de Bourgogne, Universit\'e de Bourgogne, CNRS, France}
\date{}

\usepackage{xcolor}
\usepackage{soul} 

\begin{document}

\maketitle

\begin{abstract}
We investigate a deformation of $w_{1+\infty}$ algebra recently introduced in \cite{Mago_EtAl2021.Deformed_infinity_Algebras_Celestial_CFT}  in the context of Celestial CFT that we denote by $\widetilde{W}_{1+\infty}$ algebra. 
We obtain the operator product expansions of the generating currents of this algebra and explore its supersymmetric and topological generalizations.
\end{abstract}

\tableofcontents
\vspace{1.5em}
\hrule

\section{Introduction}
Celestial holography provides new powerful tools for studying quantum gravity by applying a holographic principle to asymptotically flat spacetimes (AFS)~\cite{Pasterski_EtAl2021.Celestial_Holography}.
 It was shown in \cite{Himwich_Strominger2019.Celestial_Current_Algebra_Low_Subleading_Soft_Theorem, 
Guevara_EtAl2021.Holographic_Symmetry_Algebras_Gauge_Theory_Gravity, Donnay_EtAl2019.Conformally_Soft_Photons_Gravitons} that each of the soft theorems
 (the linear relations between scattering amplitudes of soft particles, \textit{i.e.}, the particles with energy $\omega \to 0$) in gauge theory and gravity in AFS is associated with an infinite number of symmetries and conservation laws on the celestial sphere. 
 The known symmetries do not close under commutation, implying an infinite tower of soft theorems. 
 Moreover, it has been also argued that the soft theorems in AFS are related to measurable gravitational memory effects via Fourier transform~\cite{Strominger2018.Lectures_Infrared_Structure_Gravity_Gauge_Theory, Strominger_Zhiboedov2014.Gravitational_Memory_BMS_Supertranslations_Soft_Theorems} together with a discussion towards its experimental setups~\cite{Hubner_EtAl2020.Thanks_memory_measuring_gravitationalwave_memory_first_LIGO_Virgo_gravitationalwave_transient_catalog}. 

In the recent works \cite{Guevara_EtAl2021.Holographic_Symmetry_Algebras_Gauge_Theory_Gravity,Strominger2021.infinity_Celestial_Sphere,Himwich_EtAl2022.Celestial_Operator_Product_Expansions_infinity_Symmetry_All_Spins}, it has been shown that the soft currents exhibit the so-called $w_{1+\infty}$ symmetry~\cite{popeLecturesAlgebrasGravity1991, Pope_EtAl1990.complete_structure},
where the infinite tower of soft symmetries can be organized into a single symmetry based on the wedge algebra of $w_{1+\infty}$. 
Along this direction, it has been pointed out in~\cite{Mago_EtAl2021.Deformed_infinity_Algebras_Celestial_CFT} that one can consider a more general setup, which involves non-minimal couplings in the bulk on the operator product expansions (OPEs) of soft currents in the Celestial Conformal Field Theory (CCFT).
The corresponding algebra of the currents provides a deformation of $w_{1+\infty}$ algebra, that we call $\widetilde W_{1+\infty}$ algebra.
Although $\widetilde W_{1+\infty}$ algebra bears semblance to $W_{1+\infty}$ algebra~\cite{Pope_EtAl1990.complete_structure}, it shows several specific properties that we briefly review in Section~\ref{sect-W-alg-relations}.

A purpose of this paper is to obtain a deeper understanding of algebraic aspects of $\widetilde W_{1+\infty}$ algebra.
First, we derive the OPEs for the generating currents of $\widetilde W_{1+\infty}$ algebra by generalizing the approaches used for $W_{1+\infty}$ algebra. 
This results in Theorem~\ref{th-OPE-W} in Section~\ref{OPE-paragrpah}. 
Secondly, we explore further generalization of $\widetilde W_{1+\infty}$ algebra. 
In particular, we introduce an $\mathcal{N}=2 $ supersymmetric extention of $\widetilde W_{1+\infty}$ algebra in Section~\ref{super-w-section}, and we identify a BRST operator (Theorem~\ref{th-brst-operator}) in Section~\ref{section-topological}.  
Finally, based on the OPEs for $\mathcal{N}=2$ $\widetilde W_{\infty}$ algebra together with the BRST operator, we construct the generators of the topologically twisted version $\widetilde W_{\infty}$ algebra (Theorem~\ref{th-new-generators}) in Section~\ref{super-w-section}. 
In Section \ref{conclusion-section}, we conclude by discussing potential issues and future directions.

\section{Background: soft current OPEs in CCFT} 
Let us discuss the basics of the soft current in CCFT. We follow the notation of \cite{Jiang2022.Holographic_Chiral_Algebra_Supersymmetry_Infinite_Ward_Identities_EFTs}. 
\begin{definition}[Soft current]
Let $\mathcal{O}_{h,\bar{h}}$ be a positive-helicity conformal primary celestial operator with conformal spin $s \in \frac{1}{2}\mathbb{Z}_{>0}$ and conformal dimensions $(h,\bar{h}) =  \left(\frac{k+s}{2}, \frac{k-s}{2}  \right)$, where we denote the scaling dimension by $k=s,s-1,s-2,\dots$. 	The soft current is defined by
\begin{equation}\label{soft-current} 
	H^{k,s} (z,\bar z) \coloneqq 
	\lim_{\Delta \to k} (\Delta - k) \mathcal{O}_{h, \bar h} (z,\bar z)
=	\lim_{\epsilon \to 0} \epsilon \mathcal{O}_{\frac{k+s}{2}, \frac{k-s}{2} } (z,\bar z), 
\end{equation}
where $(z, \bar z)$ is the codimension-two celestial sphere coordinate. 
\end{definition}

We start from the OPEs for these soft currents from a three-point interaction~\cite{Himwich_EtAl2022.Celestial_Operator_Product_Expansions_infinity_Symmetry_All_Spins}: 
 \begin{equation}\label{OO}
 	\mathcal{O}_{h_1, \bar h_1} (z, \bar z) 
 		\mathcal{O}_{h_2, \bar h_2} (0,0) 
 		\sim \frac{1}{z} \sum_p C_p (\bar h_1, \bar h_2) \bar z^p 
 		\mathcal{O} _{h_1+h_2-1, \ \bar h_1 + \bar h_2 +p}(0,0),
 \end{equation}
 where $p \coloneqq d_V-4$ with $d_V$ being a bulk dimension of the three-point interaction coupling particles, $  C_p (\bar h_1, \bar h_2)$ denotes the OPE coefficients and ``$\sim$'' means equality modulo expressions regular as $(z, \bar z) \to (0,0)$.  
Taking the limit $\epsilon \to 0$ as in \eqref{soft-current}, one obtain the OPE for soft currents shown in \cite{Himwich_EtAl2022.Celestial_Operator_Product_Expansions_infinity_Symmetry_All_Spins}: 
\begin{multline}
		H^{k_1, s_1} \left(z, \bar z \right)  H^{k_2, s_2} \left(w, \bar w \right) \sim -\sum_p \frac{\kappa_{s_1,s_2,-s_I}  }{2} \frac{1}{z-w}
    \sum_{\alpha=0}^{\infty} \frac{\left(\bar z- \bar w\right)^{\alpha+p}}{\alpha!} 
    \\
    \times 
    \begin{pmatrix}-2\bar{h}_1-2\bar{h}_2-2p-\alpha\\ -2\bar{h}_2-p\end{pmatrix} \bar{\partial}^\alpha H^{k_1+k_2+p-1,s_1+s_2-p-1}
     {(w,\bar w)}, 
      \label{soft-current-OPEallspins}
\end{multline}
where $\kappa_{s_1,s_2,-s_I}$ is the coupling constant of the relevant three-point amplitude, $ \bar{\partial} \coloneqq \frac{\partial}{\partial \bar w} $  and 
 $s_I = s_1+s_2-p-1$.

We use the following antiholomorphic mode expansion of the celestial currents: 
\begin{equation}\label{mode-expansion} 
	H^{k,s} (z,\bar z) = \sum_{ n = \bar h }^{-\bar h }
			H_n^{k,s} (z) \bar z ^{-n - \bar h}
			=
			 \sum_{ n =  \frac{k-s}{2}  }^{ \frac{s-k}{2}  } 
			 \frac{ H_n^{k,s} (z) }{ \bar z^{n + \frac{k-s}{2} } }, 
\end{equation} 
where each $H^{k,s}_n(z) $ is a 2D symmetry-generating conserved current whose Ward identity provides a soft theorem.\footnote{Note that \eqref{mode-expansion} holds in the MHV sector. Beyond it the polynomial becomes infinite and we should change the lower limit to $-\infty$ \cite{Mago_EtAl2021.Deformed_infinity_Algebras_Celestial_CFT}.} 
Note that we have a $z$-dependence (but not $\bar z$) for the Fourier modes of the celestial currents. 

The OPE and the commutation relations constitute an equivalent description of the corresponding algebra.
We define the following commutator~\cite{Mago_EtAl2021.Deformed_infinity_Algebras_Celestial_CFT}. 
\begin{definition}[2D-commutator]\label{definition-2D-commutator}
	\begin{align}
    \label{eq:2Dcommutator}
 \left[ H_{m}^{k_1, s_1}, H_{n}^{k_2, s_2} \right](w) \coloneqq 
 {\oint}_w \frac{\diff z}{2\pi \I }  
{\oint}_\epsilon  \frac{\diff \bar z}{2\pi \I } \bar z^{m+ \bar h_1-1}
{\oint}_\epsilon \frac{\diff \bar w}{2\pi \I } \bar w^{n+ \bar h_2-1}
H^{k_1, s_1}\left(z, \bar z \right)   H^{k_2, s_2}\left(w, \bar w \right), 
\end{align}   
where $\oint_w$ and $\oint_\epsilon$ denotes the integration over the contours around $w$ and $0$ respectively.\footnote{See \cite[App.~C]{Mago_EtAl2021.Deformed_infinity_Algebras_Celestial_CFT} for more detailed discussion.} 
\end{definition}

\begin{proposition}[Algebra of celestial currents \cite{Mago_EtAl2021.Deformed_infinity_Algebras_Celestial_CFT}] 
 The following algebraic relation holds for the celestial current modes defined in \eqref{mode-expansion},
\begin{multline}
	   \label{soft-current-commutator}
            \left[ H_{m}^{k_1, s_1}, H_{n}^{k_2, s_2} \right] (w) 
            =  - \sum_{p \, = \, \Max (s_1+s_2-3,0)}^{\Max (s_1+s_2+1,0)} 
            \frac{\kappa_{s_1,s_2,-s_I}}{2} \sum_{x=0}^p 
            \left[
            (-1)^{p-x}\binom{p}{x}\right.  
            \\ 
         \times    \left.  \binom{m+n-\bar h_1- \bar h_2-p}{m-\bar h_1-p+x}  \binom{-m-n-\bar h_1-\bar h_2-p}{-m-\bar h_1-x}  \right]
            H^{k_1 + k_2 + p - 1, s_1 + s_2 -p-1}_{m+n} (w).  
\end{multline}
For the gauge theory currents, we replace $ \frac{1}{2} \kappa_{s_1,s_2,-s_I} \mapsto \I f\indices{^{ab}_c} $ in the formula \eqref{soft-current-commutator}, where $  f\indices{^{ab}_c} $ is the structure constant, associated with the gauge symmetry.  
\end{proposition}

\section{$\widetilde W_{1+\infty}$ algebra and its structure}
\subsection{The commutation relations}\label{sect-W-alg-relations}
In the absence of higher-derivative corrections (fixed $p$ in the expansion \eqref{soft-current-commutator}), the soft currents for gravity and gauge theory obey the OPE relations of the generating currents of the $w_{1+\infty}$ algebra \cite{Guevara_EtAl2021.Holographic_Symmetry_Algebras_Gauge_Theory_Gravity,Strominger2021.infinity_Celestial_Sphere,Himwich_EtAl2022.Celestial_Operator_Product_Expansions_infinity_Symmetry_All_Spins}.
Its supersymmetric extension has been also discussed~\cite{Himwich_EtAl2022.Celestial_Operator_Product_Expansions_infinity_Symmetry_All_Spins, Jiang2022.Holographic_Chiral_Algebra_Supersymmetry_Infinite_Ward_Identities_EFTs, Ahn2022.Supersymmetric_infinity_Symmetry_Celestial_Conformal_Field_Theory}.  

Let us consider the effect of non-minimal couplings in the bulk on the OPEs of soft currents \cite{Mago_EtAl2021.Deformed_infinity_Algebras_Celestial_CFT}. We use the light transform of the celestial operator~\cite{Kravchuk_Simmons-Duffin2018.Lightray_operators_conformal_field_theory}: 
\begin{align}
    \label{lighttransform}
    {\bf \bar{L}}
    \left[\mathcal{O}_{h, \bar h } (z,\bar z) \right] =      
     \int_{\mathbb{R}} 
      \frac{\diff \bar{y} }{ 2 \pi \I }
     \frac{1}{(\bar z-\bar{y})^{2-2\bar h}} \mathcal{O}_{h,\bar h} \left(z, \bar{y}\right)  
\end{align}
to redefine the currents.
\begin{definition}[$\widetilde{W}_{1+\infty}$ currents]
 We define
	\begin{align}
    \label{w-current-def}
    \widetilde{W}^{q,s} (z, \bar z) \coloneqq \Gamma\left(2q\right) 
     {\bf {\bar{L}}}\left[ H^{s+2(1-q),s} (z,\bar z) \right],
\end{align}
where $q = 1, \frac{3}{2}, 2, \dots$ for particles with helicity $s>0$ and $\Gamma(x)$ is the Gamma function.  
\end{definition}

This implies the following mode expansion
\begin{align}\label{w-mode-expansion} 
     \widetilde{W}^{q,s}(z,\bar{z}) =  \sum_{m= \bar{h} }^{ - \bar{h} }  \widetilde W^{q,s}_{m}
     \left(z\right) \bar z^{-m- \bar {h} }
\end{align}
and  
\begin{align}\label{w-fourier-mode-def} 
 \widetilde{W}^{q,s}_{m} \left(z\right)= \left(-m+q-1\right)! \left(m+q-1\right)!H_{m}^{s+2(1-q),s} \left(z\right).
\end{align}
Using these  
currents, we can simplify the commutator \eqref{soft-current-commutator} and obtain: 
\begin{proposition}[$\widetilde{W}_{1+\infty}$ algebra~\cite{Mago_EtAl2021.Deformed_infinity_Algebras_Celestial_CFT}]
 $\widetilde{W}_{1+\infty}$ current modes obey the following algebraic relation,
	\begin{equation}  \label{Wgeneral} 
	    \left\lbrack  \widetilde{W}_{m}^{q_1,s_1} ,\,  \widetilde{W}_{n}^{q_2,s_2} \right\rbrack =
	    -\sum_{p \, = \, \Max(s_1+s_2-3,0)}^{\Max(s_1+s_2+1,0)} \frac{\kappa_{s_1,s_2,-s_I}}{2} N\left(q_1, q_2, m, n, p\right) 
                \widetilde{W}_{m + n}^{q_1+q_2-p-1, s_1+s_2-p-1}, 
\end{equation}
where $q_1,q_2 = 1, \frac{3}{2}, 2, \dots$ for particles with helicities $s_1,s_2>0 $, $m,n \in \mathbb{Z}$ and
\begin{multline}
	\label{N-def} 
	N\left(q_1, q_2, m, n, p\right)  \coloneqq 
	 \sum_{x=0}^p (-1)^{p-x} \binom{p}{x}\,  \lbrack m + q_1 - 1 \rbrack_{p-x} \  \lbrack -m + q_1 - 1 \rbrack_{x} 
	 \\ \times 
           \lbrack n + q_2 - 1 \rbrack_{x}  \ \lbrack -n + q_2 - 1 \rbrack_{p-x},
\end{multline}
with ascending and descending Pochhammer symbols 
\begin{subequations} 
\begin{align}
 (a)_n \coloneqq a(a+1) \dots (a+n-1) & = \frac{(a+n-1)!}{(a-1)!}, \label{pochhammer1} \\
 [a]_n \coloneqq a(a-1) \dots (a-n+1) & = \frac{a!}{(a-n)!}. \label{pochhammer2}
\end{align}
\end{subequations}
Note that the coefficient $N\left(q_1, q_2, m, n, p\right)$ vanishes in \eqref{Wgeneral} for specific values of $p$.
\end{proposition}

We call this $\widetilde{W}_{1+\infty}$ algebra.\footnote{The currents \eqref{w-current-def} are referred to as \textit{$W_{1+\infty}$-like} currents in~\cite{Mago_EtAl2021.Deformed_infinity_Algebras_Celestial_CFT}. 
In this regard, our $\widetilde W^{q,s}$ corresponds to their $W^{q,s}$.}
In the absense of higher-derivative corrections (fixed $p$), this algebra is reduced to $w_{1+\infty}$ algebra~\cite{Strominger2021.infinity_Celestial_Sphere,Himwich_EtAl2022.Celestial_Operator_Product_Expansions_infinity_Symmetry_All_Spins,Ahn2022.Supersymmetric_infinity_Symmetry_Celestial_Conformal_Field_Theory}.
The $\widetilde{W}_{1+\infty}$ algebra \eqref{Wgeneral} bears semblance to $W_{1+\infty}$ algebra, which is also given by a deformation of $w_{1+\infty}$ \cite{Pope_EtAl1990.complete_structure}.  
The key difference is that the truncation of $\widetilde{W}_{1+\infty}$ algebra \eqref{Wgeneral} occurs naturally, due to the absence of massless higher spin fields, while it is provided by non-trivial zeroes of the hypergeometric function $\phi^{q_1, q_2}_{2r}$ for $W_{1+\infty}$~\cite{Pope_EtAl1990.complete_structure}.

We have the coefficient function $N\left(q_1, q_2, m, n, p\right)$ for $\widetilde{W}_{1+\infty}$ algebra as in $W_{1+\infty}$ algebra.
For further calculations, we will use another representation:	
\begin{lemma}[\cite{Pope_EtAl1990.complete_structure}] 
 The coefficient function $N\left(q_1, q_2, m, n, p\right)$ is rewritten in the following form,
	\begin{multline}
	\label{N-lemma-eq} 
		N\left(q_1, q_2, m, n, p\right) = 
		 \sum_{x=0}^p (-1)^{p-x} \binom{p}{x}  \
		 [2q_2-2-x]_{p-x} \  
		 (2q_1 - 1 - p)_x  \\ \times 
		 [m+q_1 -1 ]_{p-x} \
		 [n+q_2-1]_x. 
\end{multline}
\end{lemma}

\subsection{The OPEs for $\widetilde{W}_{1+\infty}$ algebra} \label{OPE-paragrpah} 
In this Section, we provide the OPE for $\widetilde{W}_{1+\infty}$ algebra.
For this purpose, we apply the procedure described in \cite{Pope_EtAl1990.infinity_Racahwigner_Algebra} originally used for $W_{1+\infty}$ algebra.
In the context of CCFT, we have to take into account a $\bar{z}$-dependence, in contrast to the original situation discussing the chiral currents~\cite{Pope_EtAl1990.infinity_Racahwigner_Algebra}.
We apply the following steps:
\begin{enumerate}
	\item We define a new function 
	\begin{equation}
			M \left(q_1, q_2, m, n, p\right) \coloneqq 
		 \sum_{x=0}^p 
		 (-1)^{-x} 
		 \binom{p}{x}  \
		 [2q_2-2-x]_{p-x} \  
		 (2q_1 - 1 - p)_x  \
		 m^{p-x} \
		 n^x.  
	\end{equation}
	This function generates the terms of $N \left(q_1, q_2, m, n, p\right)$ of total degree $p$. 
	\item We replace the coefficient in the $\widetilde W$-commutator \eqref{Wgeneral} as follows,
		\begin{equation}
			N \left(q_1, q_2, m, n, p\right)   
			\mapsto 
			M \left(q_1, q_2,  \frac{\partial}{\partial \bar z} ,  \frac{\partial}{\partial \bar w}, p\right),
		\end{equation}
	      from which we can extract the OPE from the commutation relation of the generators.
	\item We add the factors, $(-1)$, $\frac{1}{\bar z - \bar w}$ and $\frac{1}{z - w}$ to capture the corresponding poles and comply with the Definition~\ref{definition-2D-commutator}. 
\end{enumerate} 

Here is the main result of this paper.
\begin{theorem}[OPE for  $\widetilde{W}_{1+\infty}$-currents]
\label{th-OPE-W} 
	The OPE for the $\widetilde{W}$-currents \eqref{w-current-def} has the following form:
	\begin{multline}\label{th-OPE-W-eq} 
	\widetilde W^{q_1, s_1} (z,\bar z) 
	\widetilde  W^{q_2, s_2} (w, \bar w)
	   \sim   
	   \sum_p \frac{\kappa_{s_1,s_2,-s_I}  }{2} \frac{1}{z-w}    
	   \sum_{x=0}^p 
	   (-1)^{-x}
	    \binom{p}{x} 
	   (2q_1 -1 - p)_x 
	   \\  
	   \times 
	   [2q_2 - 2 - x]_{p-x}  \ 
	   \bar{\partial}^{p-x}_{\bar z}  \ 
	     \bar{\partial}^{x}_{\bar w} 
	     \left( 
	     	\frac{\widetilde  W^{q_1+q_2-1-p, \ s_1+s_2-p-1} (w, \bar w)}{\bar z - \bar w }
	     \right).  
\end{multline}
\end{theorem} 
\begin{proof}
	We substitute the expression \eqref{th-OPE-W-eq} to Definition~\ref{definition-2D-commutator}.
	The expression takes the form
	\begin{equation}\label{W-curr-commutator-2} 
		 \left[ \widetilde  W_{m}^{q_1, s_1}, \widetilde  W_{n}^{q_2, s_2} \right](w) =
		 \oint_w \frac{\diff z}{2 \pi \I } 
{\oint}_\epsilon \frac{\diff \bar w}{2\pi \I } \bar w^{n+ q_2-1}
{\oint}_\epsilon  \frac{\diff \bar z}{2\pi \I } \bar z^{m+ q_1-1} 
\widetilde  W^{q_1, s_1}\left(z, \bar z \right)  \widetilde   W^{q_2, s_2}\left(w, \bar w \right), 
	\end{equation} 
The integral part is 
\begin{multline}\label{theorem-integral-part} 
		{\oint}_\epsilon  \frac{\diff \bar w}{2\pi \I } \bar w^{n+ q_2-1}
	\frac{\partial^x}{\partial \bar w^x}
{\oint}_\epsilon  \frac{\diff \bar z}{2\pi \I } \bar z^{m+ q_1-1} 
	\frac{\partial^{p-x} }{\partial \bar z^{p-x}  }
 \left( 
	     	\frac{\widetilde  W^{q_1+q_2-1-p, \ s_1+s_2-p-1} (w, \bar w)}{\bar z - \bar w }
	     \right)  
	     \\ = (-1)^{p-x} (p-x)!
	     	{\oint}_\epsilon  \frac{\diff \bar w}{2\pi \I } \bar w^{n+ q_2-1}
	\frac{\partial^x}{\partial \bar w^x} 
		\bigg[  \widetilde  W^{q_1+q_2-1-p, \ s_1+s_2-p-1} (w, \bar w) 
 			\\ \hspace{5cm}
 {\oint}_\epsilon  \frac{\diff \bar z}{2\pi \I } \bar z^{m+ q_1-1} 
	  \left(\bar z - \bar w \right)^{-p+x-1} \bigg] 
	   \\ = (-1)^{-1+m+q_1} (p-x)! \frac{ [-p+x-1]_{-m-q_1} }{(-m-q_1)!}
	{\oint}_\epsilon  \frac{\diff \bar w}{2\pi \I } \bar w^{n+ q_2-1} 
	   \\ \times 		\frac{\partial^x}{\partial \bar w^x}  
	   		\bigg[  \widetilde  W^{q_1+q_2-1-p, \ s_1+s_2-p-1} (w, \bar w)  \
	   		\bar w^{-p+x-1+m+q_1} \bigg].  
\end{multline}
We use a Fourier decomposition \eqref{w-mode-expansion}: 
\begin{align}
	 \widetilde  W^{q_1+q_2-1-p, \ s_1+s_2-p-1}(w,\bar{w}) &=  \sum_{j= \bar{h} }^{ - \bar{h} } 
		 \widetilde  W^{q_1+q_2-1-p, \ s_1+s_2-p-1}_j  (w) \ 
      \bar w^{-j- \bar {h} } 
      \notag \\
      &= \sum_j  \widetilde  W^{q_1+q_2-1-p, \ s_1+s_2-p-1}_j  (w) \ 
      \bar w^{-j- q_1 -q_2 +1 +p }, 
      \label{th-w-fourier-decomp} 
\end{align}
with $\bar {h}=  q_1+q_2-1-p$. Then we use an elementary formula 
\begin{equation}\label{th-elementary-formula}
	\frac{\diff^x }{\diff t^x } \left(t^n\right) = 
	\frac{n!}{(n-x)!} t^{n-x}= [n]_x \, t^{n-x}.  
\end{equation} 
for the $x$-th derivative over $\bar w $ in \eqref{theorem-integral-part}:
\begin{multline}
	\frac{\partial^x}{\partial \bar w^x} \bigg[  \widetilde  W^{q_1+q_2-1-p, \ s_1+s_2-p-1}(w,\bar{w}) \  \bar w^{-p+x-1+m+q_1}  \bigg] 
	\\
	 \mathop{=}^{\eqref{th-w-fourier-decomp} } 
	\sum_j	\frac{\partial^x}{\partial \bar w^x} \bigg[ 
	   \widetilde  W^{q_1+q_2-1-p, \ s_1+s_2-p-1}_j  (w) \ 
      \bar w^{ x+m -j  -q_2 } 	\bigg] 
      \\
      \mathop{=}^{\eqref{th-elementary-formula} } 
      	\sum_j	 
      	[x+m-j-q_2]_x \ 
      	 \widetilde  W^{q_1+q_2-1-p, \ s_1+s_2-p-1}_j  (w) \ 
      	 \bar w^{m-j-q_2}.
      	 \label{th-w-derivative}
\end{multline}
We substitute \eqref{th-w-derivative} into \eqref{theorem-integral-part}, and the integral over $\bar w $ provides 
\begin{multline}
	 \sum_j	[x+m-j-q_2]_x \ 
      \widetilde 	 W^{q_1+q_2-1-p, \ s_1+s_2-p-1}_j  (w) 
 \oint_\epsilon \frac{\diff \bar w}{2\pi \I } 
 	\bar w^{m+n - j -1 }
 	\\ 
 	=     \sum_j	[x+m-j-q_2]_x \ 
      	 \widetilde  W^{q_1+q_2-1-p, \ s_1+s_2-p-1}_j  (w)  \delta_{j, \ m+n}
      	 \\
      	= [x-n-q_2]_x 
      \widetilde 	W^{q_1+q_2-1-p, \ s_1+s_2-p-1}_{m+n}(w). 
      	\label{th-integral-over-w-bar} 
\end{multline}
The last integral is 
\begin{equation}\label{th-integral-over-z}
	\oint_w \frac{\diff z}{2\pi \I} \frac{1}{(z-w)} =1. 
\end{equation}
Bringing together \eqref{W-curr-commutator-2}, \eqref{theorem-integral-part}, \eqref{th-integral-over-w-bar}  and \eqref{th-integral-over-z}, we have
\begin{multline}
 \left[ \widetilde  W_{m}^{q_1, s_1},\widetilde   W_{n}^{q_2, s_2} \right](w)= 
	   -\sum_p \frac{\kappa_{s_1,s_2,-s_I}  }{2}    
	   \sum_{x=0}^p 
	    (-1)^{-x}
	    \binom{p}{x} 
	   (2q_1 -1 - p)_x      [2q_2 - 2 - x]_{p-x}  
	   \\ 
	      \times 
	   (-1)^{-1+m+q_1}
	    [x-n-q_2]_x   \
	   \frac{(p-x)!}{(-m-q_1)!} [-p+x-1]_{-m-q_1} 
      	\widetilde  W^{q_1+q_2-1-p, \ s_1+s_2-p-1}_{m+n}(w).  
      	\label{th-together} 
	\end{multline}
Now, using the definitions of the Pochhammer symbols \eqref{pochhammer1}-\eqref{pochhammer2} and their relation
\begin{equation}\label{pochhammer3}
	[a]_n = (-1)^n (-a)_n = (a-n+1)_n, 
\end{equation} 
we can transform:
\begin{subequations}
\begin{equation}\label{th-symbol1} 
		[x-n-q_2]_x = (-1)^x (-x+n+q_2)_x = (1-n-q_2)_x = 
	(-1)^x [n+q_2-1]_x; 
\end{equation}

\begin{align}
	\frac{(p-x)!}{(-m-q_1)!} [ -p+x-1]_{-m-q_1} &= 
	\frac{(p-x)!}{(-m-q_1)!} (-1)^{-m-q_1} (p-x+1)_{-m-q_1} 
	\notag
	\\
	& =(-1)^{-m-q_1+p-x} [m+q_1-1]_{p-x}.
	\label{th-symbol2} 
\end{align}
\end{subequations}

Substituting \eqref{th-symbol1}, \eqref{th-symbol2} to \eqref{th-together}, we arrive at the final result  
\begin{align}
	\left[\widetilde  W_{m}^{q_1, s_1}, \widetilde  W_{n}^{q_2, s_2} \right](w)&
	 \mathop{=}^{\phantom {\eqref{N-lemma-eq} } }  
	   -\sum_p \frac{\kappa_{s_1,s_2,-s_I}  }{2}    
	   \sum_{x=0}^p (-1)^{p-x} \binom{p}{x} 
	   (2q_1 -1 - p)_x  \    [n+q_2-1]_x&
	  \notag  \\ 
	  & \hspace{0.5cm}  \times  \ 
	   [2q_2 - 2 - x]_{p-x}  
	[m+q_1-1]_{p-x} 
      	\widetilde  W^{q_1+q_2-1-p, \ s_1+s_2-p-1}_{m+n}(w)
      \notag \\
      & \mathop{=}^{\eqref{N-lemma-eq}  }    -\sum_p \frac{\kappa_{s_1,s_2,-s_I}  }{2} 
      N\left(q_1, q_2, m, n, p\right) 
      \widetilde 	W^{q_1+q_2-1-p, \ s_1+s_2-p-1}_{m+n}(w).
  \end{align}
\end{proof}

\section{Supersymmetric and topological generalizations} 

The $\mathcal{N}=2$ supersymmetric $W_{\infty}$ algebra can be twisted to construct a topological $W_\infty$ algebra \cite{Pope_EtAl1991.Wtopological_matter_gravity}, which has been also discussed in the context of CCFT in particular for $w_\infty$ algebra~\cite{Ahn2022.Supersymmetric_infinity_Symmetry_Celestial_Conformal_Field_Theory}. 
Here we apply a similar construction for $\widetilde{W}_{1+\infty }$ algebra \eqref{Wgeneral}. For this purpose, we first construct an $\mathcal{N}=2$ supersymmetric extension of $\widetilde{W}_\infty$ algebra.\footnote{The notation ``super $\widetilde{W}_{\infty}$'' (and not ``$1+\infty$'') comes from the fact that the corresponding supersymmetric algebra should contain both $\widetilde{W}_{1+\infty} \times \widetilde{W}_\infty $ as a bosonic sector.}

\subsection{$\mathcal{N}=2$ supersymmetric extension} \label{super-w-section}
The general structure of the $\mathcal{N}=2$ super $W_{\infty}$ algebra can be schematically drawn as follows~\cite{Bergshoeff_EtAl1990.super_algebra}: 
\begin{align}  
   WW&\sim W       & W  \tilde W & \sim 0             &  G^\pm G^\pm  &\sim 0 \notag \\ \label{structure} 
\tilde W \tilde W  &\sim \tilde W        &  G^{\pm} G^{\mp }  &  \sim W \oplus \tilde W   &    &  \\
W G^{\pm }  &\sim G^{\pm}     & \tilde W G^\pm    &\sim G^\pm          & &   \notag
\end{align}
Here $W$ are $W_\infty$ currents, $\tilde W $ are $W_{1+\infty}$ currents ($ {W}_{1+\infty} \times {W}_\infty $ subalgebras constitute a bosonic sector of the new superalgebra), and $G^\pm $ are fermionic currents. 
In order to have this structure for $\widetilde{W}_\infty$ algebra, we define the fermionic currents $G^{q+}$ and $G^{q-}$:  
\begin{multline} 
		\{ G^{q_1- }_m,  G^{q_2 + }_n
	\} =- \sum_p \frac{\kappa_{s_1, s_2, -s_I} }{2} N(q_1, q_2, m,n,p) 
	\big( b_ {q_1, q_2,x}^{m,n,p} 
	\widetilde W^{q_1 +q_2 -p -1 }_{m+n}  
	+ (-1)^p 
	\tilde 	b_{q_1,q_2,x}^{m,n,p} 
	 \doublewidetilde{W}{}^{q_1 + q_2 -p -1 }_{m+n}  \big),  \label{{G,G}} 
\end{multline}
with additional structure constants $b_ {q_1, q_2, x}^{m,n,p}$ and $  \tilde 	b_{q_1, q_2, x}^{m,n,p} $ (unknown at the moment). Here $\widetilde W$ are  $\widetilde W_\infty$ currents, $  \doublewidetilde{W} $ are $\widetilde W_{1+\infty}$ currents. For the fermionic current $G^{q \pm}$, we require $s \in \mathbb{Z}_{> 0} + \frac{1}{2}$ so that they generate the fermionic sector of the algebra.
 Using the~Theorem~\ref{th-OPE-W}, we may switch \eqref{{G,G}} to the OPE form: 

\begin{corollary} 
 The OPE of the fermionic currents of $\mathcal{N}=2$ $\widetilde{W}_\infty $ algebra is given by
	\begin{multline}  \label{G-OPE} 
	G^{q_1 - } (z,\bar z) 
	G^{q_2 + } (w, \bar w)
	   \sim      \sum_p \frac{\kappa_{s_1,s_2,-s_I}  }{2} \frac{1}{z-w}    
	   \sum_{x=0}^p 
	   (-1)^{-x}
	    \binom{p}{x} 
	   (2q_1 -1 - p)_x  
	    \\*   \times [2q_2 - 2 - x]_{p-x}    
	  \bar{\partial}^{p-x}_{\bar z}  \ 
	     \bar{\partial}^{x}_{\bar w} \bigg[
	    \frac{1}{(\bar z - \bar w)}
	     \bigg( 
	     B_{q_1, q_2}^{p,x}
	     	\widetilde W^{  q_1+  q_2-1-p}{(w, \bar w)}
	     	\\*   
	     	+ (-1)^p 
	     	\tilde B_{q_1, q_2}^{p,x} 
	      \doublewidetilde{W}{}^{q_1+q_2-1-p}{(w, \bar w)} 
	     \bigg)  \bigg],
\end{multline}
where the new coefficients $B_{q_1, q_2}^{p,x}$ and $\tilde B_{q_1,q_2}^{p,x}$ which are independent of the Fourier degrees $(m,n)$, are determined from  $b_ {q_1,q_2,x}^{m,n,p}$ and $  \tilde 	b_{q_1, q_2,x}^{m,n,p} $  in the same way via the function $M \left(q_1, q_2,  \frac{\partial}{\partial \bar z} ,  \frac{\partial}{\partial \bar w}, p\right)$ from $N\left(q_1, q_2, m, n, p\right)$ as before.
\end{corollary}

In order to calculate some particular values of the structure constants $B_{q_1, q_2}^{p,x}$ and $\tilde B_{q_1,q_2}^{p,x}$  in \eqref{G-OPE}, we may use explicit realizations of our currents. Let us introduce ghost fields with the following OPE:  
\begin{equation}\label{ghost-OPE}
	b_i ( z, \bar z) c_j (w, \bar w) \sim 
	\frac{1}{z-w}
	 \frac{\delta_{ij} }{\bar z - \bar w}. 
\end{equation}
Then we introduce an explicit realization of the $\widetilde{W}_{1+\infty} $ currents \eqref{th-OPE-W-eq} with $ q \geq 2$: 
\begin{multline}
	 \label{realization-intro-1}
 \widetilde W^{q,s} (z, \bar z) \coloneqq  \sum_{k \geq 0 }  \kappa_{s_1,s_2,-s_0}  
 (q +2) \colon \bar \partial_{\bar z} c_k (z,\bar z) b_{q-1+k}  (z, \bar z) \colon  
 \\  + \kappa_{s_1,s_2,-s_0}    (k+1) \colon c_k (z, \bar z)
  \bar  \partial_{\bar z} b_{q-1+k} (z, \bar z) \colon 
 + \dots,  
\end{multline}
where  $s_0 \coloneqq s_I (p=1) = s_1 + s_2 -2 $, $\colon \cdot \,  \colon $ means normal-ordered product and $\dots$ states for the terms of the form $ \colon \bar \partial^\alpha_{\bar z}  c_k (z,\bar z) \bar  \partial^\beta_{\bar z} b_{q-1+k}  (z, \bar z) \colon$  with total degree $\alpha + \beta \geq 2 $. Similarly, we use the anticommuting ghosts $\tilde b_i (z,\bar z) $ and $\tilde c_j (w,\bar w)$ to construct a realization of the $\widetilde{W}_\infty  $ algebra, where we denote the currents by $\doublewidetilde{W}{}^{q,s} (z, \bar z) $.

Using Wick's theorem and \eqref{ghost-OPE}, one can check that the currents $w^{q} $, for which the terms given in \eqref{realization-intro-1} constitute the complete expression, \textit{i.e.} 
\begin{equation}
	 w^q (z, \bar z) \coloneqq 
	 - \sqrt{  \kappa_{s_1,s_2,-s_0}  }  \sum_{k \geq 0 } 
 (q +2) \colon 
 \bar \partial_{\bar z} c_k (z,\bar z) b_{q-1+k}  (z, \bar z) \colon  
 + (k+1) \colon c_k (z,\bar z) 
\bar  \partial_{\bar z} b_{q-1+k} (z,\bar z) \colon  
\end{equation}
generate the $w_{1+\infty} $ algebra \cite{Pope_EtAl1991.Conditions_anomalyfree_superW_algebras}\footnote{Here, no central term is considered.}: 
\begin{multline}
	\label{realization-ope1} 
	w^{q_1} (z,\bar z) w^{q_2} (w, \bar w) \sim -  \kappa_{s_1,s_2,-s_0}
	 (q_1 + q_2 -2)  \frac{ w^{q_1 + q_2 - 2  } (w,\bar w) }{ (\bar z - \bar w)^2 } 
	\\ -   \kappa_{s_1,s_2,-s_0}
	(q_1 -2) 
	\frac{\partial w^{q_1+q_2 -2}  (w,\bar w) }{ \bar z - \bar w }. 
\end{multline}
The OPE \eqref{realization-ope1} coincides with the main OPE \eqref{th-OPE-W-eq} in $p=0$ approximation (with no higher-derivative corrections). Hence, in general, the currents \eqref{realization-intro-1} generate the algebra \eqref{th-OPE-W-eq}. The expression for the explicit realization is given by 
\begin{equation}\label{realization-with-alpha-consts} 
	 \widetilde W^{q,s} (z, \bar z ) \coloneqq 
	 \sum_{k, n \geq 0 \atop \alpha + \beta = n  } \alpha_{\alpha, \beta }^{q, k}
	  \colon  
	  \bar \partial^\alpha_{\bar z}  c_k (z,\bar z) 
	  \bar \partial^\beta_{\bar z} b_{q-1+k}  (z,\bar z) \colon,
\end{equation}
where 
\begin{align}
	\alpha^{q,k}_{0,1} & = -  
	\sqrt{ \kappa_{s_1,s_2,-s_0} }
	 (q+2), 
	\\
	 \alpha^{q,k}_{0,1} & =-   
	\sqrt{ \kappa_{s_1,s_2,-s_0} }
	 (k+1), 
	 \\
	  &    \dots \notag 
\end{align} 

Now, in addition to the anticommuting ghosts $b_i (z, \bar z)$, $c_j (z, \bar z) $ for $\widetilde{W}^{q,s} $ and $\tilde b_i (z,\bar z)$, $\tilde c_j (z,\bar z) $ for $\doublewidetilde{W}{}^{q,s} $, we introduce the commuting ghosts $\beta_i (z,\bar z)$, $\gamma_j (z,\bar z) $ for $G^{q -} (z, \bar z)$ and  $\bar \beta_i (z,\bar z)$, $\bar \gamma_j (z,\bar z) $ for $G^{q +} (z, \bar z)$ \cite{Pope_EtAl1991.Conditions_anomalyfree_superW_algebras}: 
	\begin{subequations}
		\begin{align}
			\bar \beta_i (z,\bar z) 
			\gamma_j  (w,\bar w)  
		 &	\sim 
		\beta_i (z, \bar z)  \bar 	\gamma_j   (w,\bar w) \sim \frac{1}{z-w} 
			 \frac{\delta_{ij } }{\bar z - \bar w }, 
			 \\ 
			 \bar \gamma_i  (z,\bar z) 
			 \beta_j (w,\bar w)  
		&	\sim 
		  	\gamma_i   (z,\bar z)
		  	\bar \beta_j (w,\bar w)
		  	 \sim 
		- \frac{1}{z-w} 	 \frac{\delta_{ij } }{\bar z - \bar w }. 
		\end{align}
	\end{subequations}

Using this, one can calculate 
\begin{align}
		G^{q_1 - } (z,\bar z) 
	G^{q_2 + } (w, \bar w) \sim 
	 \frac{2 \kappa_{s_1,s_2,-s_0}  }{ \bar z - \bar w } \widetilde{W}^{ q_1 +q_2 -1 }
	&-2  \kappa_{s_1,s_2,-s_0}  \frac{ (q_1 +q_2 -2 )  }{(\bar z - \bar w  )^2}  \doublewidetilde{W}{}^{ q_1 + q_2 -1 } (w, \bar w) \notag 
	\\ &-2  \kappa_{s_1,s_2,-s_0}  \frac{(q_1 -1  )}{(\bar z - \bar w)} 
	\frac{\partial}{\partial \bar w} 
	\doublewidetilde{W}{}^{ q_1 +q_2 -1 } (w, \bar w) 
	+ \dots 	\label{GG-real}
\end{align}
or, in the anticommutator form (using mode expansions): 
\begin{align}
	 	\{ G^{q_1 - }_m,  G^{q_2 + }_n
	\}
	= 2  \kappa_{s_1,s_2,-s_0}   \widetilde{W}^{  q_1 +q_2 -1 }_{m+n} 
	- 2  \kappa_{s_1,s_2,-s_0}  [ m ( q_2  - 1 ) - n (q_1 - 1 )  ] 
 \doublewidetilde{W}{}^{ q_1 + q_2 -1 }_{m+n}
 + \dots 
\end{align}
By comparing \eqref{GG-real} and \eqref{G-OPE}, we can finally determine some values of the structure constants  $B_{q_1, q_2}^{p,x}$ and $\tilde B_{q_1,q_2}^{p,x}$ and conditions for them: 
\begin{subequations}
	\begin{align}
		B^{1,0}_{q_1,q_2} &= 0, \label{B-const-a} 
		\\
		\tilde B^{1,0}_{q_1,q_2} &= -2
		\cdot 
		\frac{(2q_1+q_2 - 3 )}{ (q_2 - 1)}, 
		\\
		B^{1,1}_{q_1,q_2} & = 0, 
		\\
		\tilde B^{1,1}_{q_1,q_2} &= -2, 
		\label{B-const-d} 
		\\
		& \dots \notag 
	\end{align}
\end{subequations}
 Further structure constants may be obtained by the following algorithm. First, one needs to build a realization of the higher-order terms of the algebra OPE. For this purpose, it would be helpful to consider terms with $\alpha + \beta \geq 2 $ in  \eqref{realization-with-alpha-consts}. After  OPE calculation $\widetilde W^{q_1, s_1} (z, \bar z) \widetilde W^{q_2, s_2} (w, \bar w) $, using Wick's theorem and expanding, one can compare the obtained expression with our result \eqref{th-OPE-W-eq}, and thereby determine $\alpha^{q,k}_{\alpha,\beta} $ for some further $\alpha, \beta $. Then, using the obtained $\alpha^{q,k}_{\alpha,\beta}$, one can similarly construct  realizations of the fermionic currents $G^{q,s\pm }$ and repeat the procedure to calculate the structure constants $B^{p,x}_{q_1,q_2} $ and $\tilde B^{p,x}_{q_1,q_2} $. 

Here we calculated some particular values of the structure constants \eqref{B-const-a}-\eqref{B-const-d} and provided an algorithm to do it for the cases $p,x>1 $, while its practical implementation might require more effort from the calculation point of view.

\subsection{Topological twist} \label{section-topological}
The first step towards the topological twist is to define the BRST operator: 
\begin{theorem}[BRST operator for super $\widetilde W_\infty$] 
	\label{th-brst-operator}
	The BRST operator for the super $\widetilde W_\infty$ algebra is given by 
	\begin{equation}
		Q = G^{  3/2 + }_{-1/2 } (z), 
	\end{equation}
	where the lower index numerates the Fourier mode coefficient.
\end{theorem}
\begin{proof}
 The BRST operator is given by the contour integral of the fermionic current,
	\begin{align}
	Q = \oint_\epsilon \frac{\diff \bar z}{2\pi \I } G^{3/2+}(z, \bar z) &=
	 \sum_r  \oint_\epsilon \frac{\diff \bar z}{2\pi \I } G^{3/2+}_r (z) \ \bar z^{-r - 3/2 } 
	 \notag \\ 
	 &= \sum_r G^{3/2+}_r (z)\delta_{r,-1/2} =
	 G^{3/2+}_{-1/2} (z). 
\end{align}   
\end{proof}

We again note that we have an additional $\bar w$-dependence for celestial currents. 
Now let us proceed to determine the generators of the topological $\widetilde{W}_\infty$ algebra. 

\begin{theorem}[Generators of topological $\widetilde W_\infty$ algebra] \label{th-new-generators} 
	We define the new generators by 
\begin{equation}\label{def-topological-generators} 
	\hat V^{q} (w,\bar w) \coloneqq \left\{ Q,  G^{q-}(w, \bar w)  \right\},
\end{equation}
which are BRST-exact and hence generate a topological algebra. 
	They are given by  
 \begin{align}
	\hat V^{q} (w,\bar w)  =   \sum_p \frac{\kappa_{s_1,s_2,-s_I}  }{2} \frac{1}{z-w} 
	\frac{\partial^p}{\partial \bar w^p} 
	 \bigg(   (-1)^p   
	     B_{3/2, q}^{p,p}
	     	\widetilde W^{-\frac{1}{2} +q -p}{(w, \bar w)}
	\notag   \\  	
	+  \tilde B_{3/2, q}^{p,p}
	     	\doublewidetilde{W}{}^{-\frac{1}{2} +q-p}{(w, \bar w)} 
	     \bigg). 
\end{align}

\end{theorem} 
\begin{proof} 
The definition \eqref{def-topological-generators} can be re-written as follows: 
\begin{equation}
	\hat V^{q} (w,\bar w) = \oint_{\bar w} \frac{\diff \bar z}{2 \pi \I }
	G^{3/2+} (z,\bar z) G^{q-} (w, \bar w), 
\end{equation} 
To obtain $G^{3/2+} (z,\bar z) G^{q-} (w, \bar w) $, we use the OPE \eqref{G-OPE} for the case $q_1 = \dfrac{3}{2} $,  $q_2 \eqqcolon  q \in \mathbb{Z}/2$ and take the derivative over $ \bar z$: 
\begin{multline}
G^{3/2 + } (z,\bar z) 
	G^{q - } (w, \bar w)
	   \sim      \sum_p \frac{\kappa_{s_1,s_2,-s_I}  }{2} \frac{1}{z-w}    
	   \sum_{x=0}^p 
	    \binom{p}{x} 
	   (2- p)_x  [2 q - 2 - x]_{p-x}  (p-x)!
	   \\ 
	 \times  
	     \frac{\partial^x}{\partial \bar w^x}
	     \bigg[
	 (\bar z - \bar w)^{-p+x-1}
	     \bigg(
	     (-1)^p   
	     B_{3/2, q}^{p,x}
	     	\widetilde W^{-\frac{1}{2} +q -p}{(w, \bar w)}
	     	+   
	     	\tilde B_{3/2, q}^{p,x}
	     	\doublewidetilde{W}{}^{-\frac{1}{2} +q-p}{(w, \bar w)} 
	     \bigg)  \bigg].   	
\end{multline}
The integration over $\bar z$ 	gives: 
\begin{equation}
	\oint_{\bar w} \frac{\diff \bar z}{2\pi \I } \frac{1}{(\bar z - \bar w)^{p -x +1}} = \delta_{p,x}
\end{equation}
This allows us to remove the sum over $x$ and obtain 
\begin{align}
	\hat V^{q} (w,\bar w)  =   \sum_p \frac{\kappa_{s_1,s_2,-s_I}  }{2} \frac{1}{z-w} 
	\frac{\partial^p}{\partial \bar w^p} 
	 \bigg(   (-1)^p   
	     B_{3/2, q}^{p,p}
	     	\widetilde W^{-\frac{1}{2} +q -p}{(w, \bar w)}
	\notag   \\  	
	+  \tilde B_{3/2, q}^{p,p}
	     	\doublewidetilde{W}{}^{-\frac{1}{2} +q-p}{(w, \bar w)} 
	     \bigg). 
\end{align}
\end{proof}

Now let us discuss the algebraic relations. We re-write \eqref{def-topological-generators} in terms of the Fourier modes 
\begin{equation}\label{def-topological-generators-fourier}
		\hat V^{q}_m (w,\bar w) \coloneqq \left\{ Q,  G^{q-}_{m+1/2} 
		(w)  \right\},
\end{equation}
to compute the following commutator: 
\begin{align}
	[ \hat V^{q_1}_m, \hat V^{q_2}_n ]  & \mathop{=}^\eqref{def-topological-generators-fourier}  
		\left[ 
	\left\{ G^{3/2+}_{-1/2} (z), \ G^{q_1 -}_{m+1/2} (z) 
	\right\},
	\ 
	\left\{ G^{3/2+}_{-1/2} (z), \ G^{q_2 -}_{n+1/2} (z) 
	\right\} 
	\right] 
	\notag \\
	&  \mathop{=}^{\phantom{\eqref{def-topological-generators-fourier}}}   
	\left\{ G^{3/2+}_{-1/2} (z), \ \left[
		G^{q_2 -}_{n+1/2} (z), \
		\left\{ G^{3/2+}_{-1/2} (z), \ G^{q_1 -}_{m+1/2} (z) 
		\right\} 
	\right]
	\right\} \notag
	\\ 
	& \mathop{+}^{\phantom{\eqref{def-topological-generators-fourier}}} 
	\left\{ G^{q_2 -}_{n+1/2} (z), \ 
		\left[ 
		\left\{ 
		G^{3/2+}_{-1/2} (z),  \ 	G^{q_1 -}_{m+1/2} (z)
		\right\}, \  
		G^{3/2+}_{-1/2} (z)
		\right]   
	\right\} 
	\notag
	\\
	&  \mathop{=}^{\phantom{\eqref{def-topological-generators-fourier}}}   
	\left\{ G^{3/2+}_{-1/2} (z),  \ 
	\left[G^{q_2 -}_{n+1/2} (z), \ 
	\left\{ 
	G^{3/2+}_{-1/2} (z),  \  G^{q_1 -}_{m+1/2} (z)
	\right\} 
	\right] 
	\right\},  \label{commutators-jacobi}
\end{align}
where we used the Jacobi identity. 
The conditions for the Jacobi identity to hold in the case of $\widetilde W_{1+\infty}$ 
were obtained in \cite{Mago_EtAl2021.Deformed_infinity_Algebras_Celestial_CFT}: 
\begin{align}
	\frac{\kappa_{0,1,1}}{\kappa_{-2,2,2}} = 
	\frac{\kappa_{1,1,2}}{\kappa_{0,2,2}}; 
	\quad  
	\frac{\kappa_{-1,1,1}}{\kappa_{-1,1,2}} = 
	\frac{\kappa_{1,1,1}}{3\kappa_{1,1,2}};
	\quad 
	\kappa_{-1,1,1}=\kappa_{0,0,1};
	\quad 
	\kappa^2_{0,1,1}=2\kappa_{1,1,1}\kappa_{-1,1,1}.
\end{align}
We assume that the conditions are the same for the super $\widetilde W_\infty $, hence the Jacobi identity holds here as well. However, it might be possible to refine these conditions using the provided explicit realizations \eqref{realization-with-alpha-consts} in future work. 

Hence from \eqref{structure} we see that the right-hand side of \eqref{commutators-jacobi} leads to $\hat V_l^{q} $ again.
We can write 
\begin{equation}
		[ \hat V^{q_1}_m, \hat V^{q_2}_n ] = \sum_{p} \hat g (q_1, q_2, m,n, p) \hat V^{q_1 + q_2-p-1}_{m+n}, 
\end{equation}
where structure constants $\hat g$ could be expressed in terms of  $b_ {q_1,q_2,x}^{m,n,p}$ and $  \tilde 	b_{q_1, q_2,x}^{m,n,p} $ from \eqref{{G,G}}. 

In the same way, we can obtain other relations for the generators, and thus the full structure of topological $\widetilde W_{\infty}$ algebra is given in the form:  
\begin{subequations}
\begin{align}
	[ \hat V^{q_1}_m, \hat V^{q_2}_n ] &= \sum_{p} \hat g (q_1, q_2, m,n, p) \hat V^{q_1 +q_2-p-1}_{m+n},  \label{topological1} 
	\\
	[ \hat V^{q_1}_m, \hat G^{q_2}_{n+1/2} ] &= \sum_{p} \hat g (q_1, q_2, m,n, p) \hat G^{q_1 + q_2-p-1}_{m+n+1/2}, \label{topological2} 
	\\
	\{ G^{q_1}_{m+1/2}, G^{q_2}_{n+1/2} \} & =0. 
	\label{topological3} 
\end{align}
\end{subequations} 
 
Although determining the structure constants $\hat g$ is difficult in general, we are able to compare them with a specific example.
It has been known that $\widetilde{W}_\infty$ is reduced to $w_\infty$ as follows:
rescaling the generators
\begin{equation}
	v^{q}_m \to \lambda^{ q-2} \hat V^{q}_m, 
	\quad 
	G^{q}_m \to \lambda^{q-2}  \hat G^{q}_m, 
\end{equation}
and then taking the limit $\lambda \to 0$ with $p=1$. 
Hence, from \eqref{topological1} and \eqref{topological2}, we obtain 
\begin{subequations}
\begin{align}
	[ v^{q_1}_m, v^{q_2}_n ] &=
	 \hat g (q_1,q_2, m,n, 1)  v^{q_1 + q_2-2}_{m+n}, 
	 \\
	 [ v^{q_1}_m,  \hat G^{q_2}_{n+1/2} ] &= 
	 \hat g (q_1, q_2, m,n, 1) \hat G^{q_1 + q_2-2}_{m+n+1/2}.
\end{align}
\end{subequations}
Compared with the results obtained in \cite{Ahn2022.Supersymmetric_infinity_Symmetry_Celestial_Conformal_Field_Theory}, we conclude that 
\begin{equation}\label{particular-const} 
	\hat g (q_1, q_2, m,n, 1) = m (q_2-1) - n (q_1 - 1). 
\end{equation}

\section{Conclusion and outlook}
\label{conclusion-section} 
In this work, we investigated a deformation of $w_{1+\infty}$ that we call $\widetilde{W}_{1+\infty}$ algebra \eqref{Wgeneral}, which was previously obtained in \cite{Mago_EtAl2021.Deformed_infinity_Algebras_Celestial_CFT} in the context of CCFT.
This algebra, 
despite being different from the known $W_{1+\infty}$, this algebra shares several common properties, which allows us to apply analogous approaches to study its algebraic structure. 
In particular, we obtained the OPE for $\widetilde W_{1+\infty}$-algebra (Theorem~\ref{th-OPE-W}, Section \ref{OPE-paragrpah}). 

We constructed the BRST operator of $\mathcal{N}=2$ supersymmetric $\widetilde W_{\infty}$-algebra (Theorem~\ref{th-brst-operator}) and then obtained the generators of the topological $\widetilde W_{\infty}$-algebra (Theorem~\ref{th-new-generators}) in Section~\ref{section-topological}. We determined the structure constants $B_{q_1, q_2}^{p,x} $ and $\tilde B_{q_1, q_2}^{p,x} $ in \eqref{G-OPE} in several particular cases and outline the possible algorithm of treatment the others. 
Along this direction,  calculating the structure constants  in more cases (by, for example, practical implementation of the provided in \ref{super-w-section} algorithm), or even obtaining a closed-form expression might be a subject of future work.

It would be also important to revisit the issue of physical interpretation of the supersymmetric and topological generalizations, presented in this paper in the context of CCFT and compare them to the known results concerning the interpretation of $W$-symmetry in CCFT in the non-commutative setup  \cite{Monteiro_2023,Bu_EtAl2022.Moyal_deformations_1_infty_celestial_holography, guevara2022gravity, Bittleston_EtAl2023.Celestial_Chiral_Algebra_SelfDual_Gravity_EguchiHanson_Space}. 

\section*{Acknowledgements}
P.\,D. acknowledges the financial support from Math4Phys program of Universit\'e de Bourgogne \textit{\&} Universit\'e Bourgogne Franche-Comt\'e. This work was in part supported by EIPHI Graduate School (No.~ANR-17-EURE-0002) and Bourgogne-Franche-Comt\'e region.

\bibliographystyle{amsalpha_mod}
\bibliography{bibliography_bibtex}

\end{document}